\documentclass[runningheads]{llncs}

\usepackage{amsmath}
\usepackage{amssymb}
\usepackage{cite}
\usepackage{pifont}
\usepackage{stmaryrd}
\usepackage{graphicx}
\usepackage{latexsym}
\usepackage{enumerate}
\usepackage{subfigure}
\usepackage[ruled,vlined,linesnumbered]{algorithm2e}
\DontPrintSemicolon

\newcommand{\NP}{\ensuremath{\mathcal{NP}}}

\begin{document}

\titlerunning{1-visibility }

\title {1-Visibility Representations of 1-Planar Graphs}

\author{Franz J. Brandenburg
\thanks{This work is supported by the German Science Foundation, DFG, Grant
 Br-835/18-1}}
  \institute{University of Passau, 94030 Passau, Germany \\
  \email{brandenb@informatik.uni-passau.de}}

\maketitle

\begin{abstract}
A visibility representation is a classical drawing style of planar
graphs. It displays the vertices of a graph as  horizontal
vertex-segments, and each edge is represented by a vertical
edge-segment touching the segments of its end vertices; beyond that
segments do not intersect.

We generalize visibility to 1-visibility, where each edge- (vertex-)
segment crosses at most one vertex- (edge-) segment. In other words,
a vertex is crossed by at most one edge, and vice-versa. We show
that 1-visibility properly extends 1-planarity and develop a linear
time algorithm to compute a 1-visibility representation of an
embedded 1-planar graph on $\mathcal{O}(n^2)$ area. A graph is
1-planar if it can be drawn in the plane such that each edge is
crossed at most once.
 Concerning  density, both 1-visible and
1-planar graphs of size $n$ have at most $4n-8$ edges. However, for
every $n \geq 7$  there are 1-visible graphs with $4n-8$ edge which
are not 1-planar.

\end{abstract}


\section{Introduction}

 Drawing planar graphs is an important topic in graph theory,
 combinatorics, and in particular in graph drawing. The existence of
straight-line drawings was independently proved by Wagner
\cite{w-bv-36}, Steinitz and Rademacher \cite{sr-34}, Stein
\cite{s-cm-51} and F\'{a}ry \cite{fary-48}. The stunning results of
de Fraysseix, Pach and Pollack \cite{fpp-hdpgg-90} and Schnyder
\cite{S-epgg-90} show that planar graphs admit straight-line grid
drawings in quadratic area, which can be computed in linear time.

A visibility representation is another way to draw a planar graph.
Here the vertices are drawn as horizontal bars, called
\emph{vertex-segments}, and two vertex-segments must see each other
along a vertical line, called \emph{edge-segment}, if there is an
edge between the respective vertices. Vertex- and edge-segments do
not overlap except that an edge-segment begins and ends at the
vertex-segments of its end vertices. For convenience, we identify
vertices and edges with their segments. Otten and van Wyck
\cite{ow-grild-78} showed that every planar graph has a visibility
representation, and a linear time algorithm for constructing it was
given independently by Rosenstiehl and Tarjan \cite{rt-rplbopg-86}
and by Tamassia and Tollis \cite{TT-vrpg-86}. Their algorithm uses a
grid of size at most $(2n-5) \times (n-1)$, which was gradually
improved to
 $(\lfloor 4n/3 \rfloor -2)  \times (n-1)$
 \cite{flly-wovrpg-07}.

 Visibility representations lead to clear pictures and have gained a
 lot of interest, see also \cite{dett-gdavg-99} and the references given there.
 For planar
graphs there are  three (main) versions of visibility: weak,
$\epsilon$, and strong.  They differ in the representation of the
segments. In the weak version the  vertex-segments may or may not
include their extremes, such that an edge-segment  may pass the open
end of a third unrelated one. Two vertices must see each other if
they are adjacent, but not conversely. Hence, weak visibility
preserves the \emph{subgraph property}, which says that every
subgraph of a weak visibility graph is a weak visibility graph. In
the $\epsilon$-version, the edge-segments are bands with thickness
$\epsilon
> 0$ and there is an edge if and only if the corresponding vertices
see each other.  Finally, in the strong version there is an edge if
and only if there is a visibility. The latter makes an essential
difference, since the $K_{2,3}$ has no strong visibility
representation. Moreover, it is \NP-hard to determine whether a
3-connected planar graph has a strong visibility representation
\cite{a-rvg-92}, whereas weak visibility is equivalent to planarity
and thus testable in linear time. Weak and $\epsilon$-visibility
coincide on 2-connected planar graphs, and every 4-connected planar
graph has a strong visibility representation \cite{TT-vrpg-86}.

There are several attempts to generalize the planar graphs to nearly
planar graphs, e.g., via forbidden minors \cite{rs-gmV-86}, surfaces
of higher genus, or various restrictions on crossings, such as
$k$-planar \cite{pt-gdfce-97}, almost planar \cite{FPS-nekqpg-13} or
right angle crossing (RAC) graphs \cite{del-dgrac-11}. Here, we
consider 1-planar graphs, which are defined by drawings in the plane
such that each edge is crossed at most once. 1-planar graphs were
introduced by Ringel \cite{ringel-65} and occur when a planar graph
and its dual are drawn simultaneously \cite{ek-sepgd-05}. As an
example consider the complete graph $K_6$, which can be drawn
1-planar with two nested triangles and straight-line edges.

The straight-line or rectilinear drawability of 1-planar graphs was
first investigated by Eggleton \cite{e-rdg-86}. He settled this
problem for outer 1-planar graphs and proved that every outer
1-planar graph has a straight-line drawing. In outer 1-planar graphs
all vertices are in the outer face and each edge is crossed at most
once. Thomassen \cite{t-rdg-88} generalized this result and proved
that an embedded 1-planar graph has a straight-line  drawing if and
only if it excludes B- and W-configurations, see Figs.
\ref{fig:B-configuration} and \ref{fig:W-configuration}. Then only
X-configurations remain for pairs of crossing edges, see Fig.
\ref{fig:X-configuration}. The forbidden configurations were
rediscovered by Hong et al. \cite{help-ft1pg-12}, who also showed
that there is a linear time algorithm to convert a 1-planar
embedding without these forbidden configurations into a
straight-line drawing. In fact, 1-planar graphs can be drawn as nice
as planar graphs. Alam et al. \cite{abk-sld3c-13} proved that every
3-connected 1-planar graph has an embedding with at most one
W-configuration in the outer face, and has a straight-line grid
drawing in quadratic area with the exception of a single edge in the
outer face. Such  drawings can be computed in linear time from a
given 1-planar embedding as a witness for 1-planarity. Here we add
visibility representations.

\begin{figure}
  \centering
  \subfigure[B-configuration]{
      \parbox[b]{3.0cm}{
      \centering
        \includegraphics[scale=0.4]{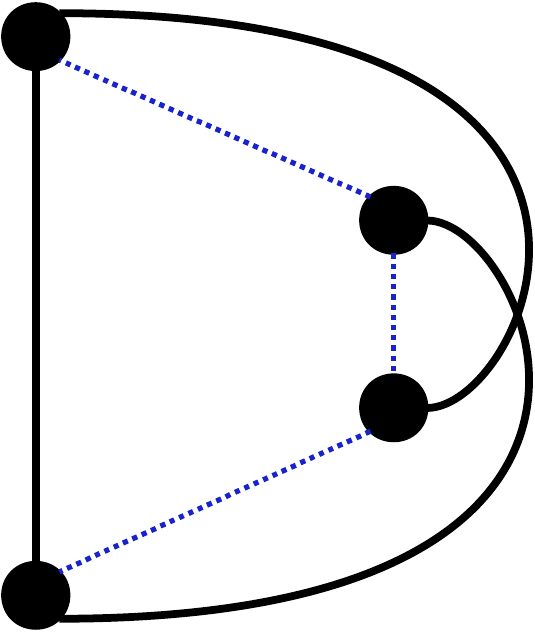}
        \label{fig:B-configuration}
      }
  }\hfil
  \subfigure[W-configuration]{
    \parbox[b]{3.0 cm}{
    \centering
    \includegraphics[scale=0.42]{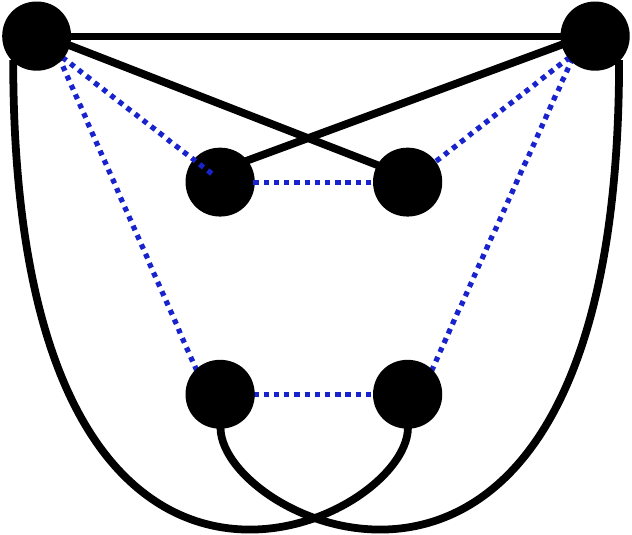}
    \label{fig:W-configuration}
  }
  }
\hfil
  \subfigure[X-configuration]{
  \parbox[b]{3.0cm}{
  \centering
    \includegraphics[scale=0.43]{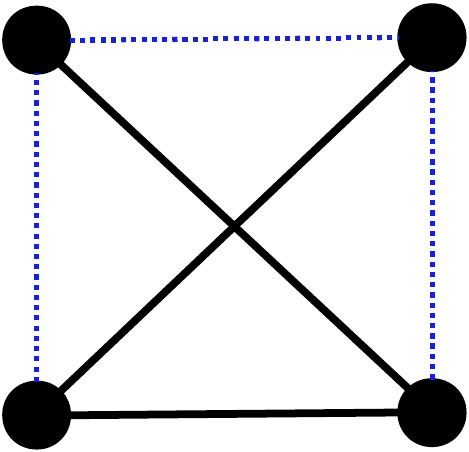}
    \label{fig:X-configuration}
  }
  }
  \caption{Augmented B-, W- and X-configurations, where the augmentation
  is drawn with blue dotted lines.
  \label{fig:BWXconfig}}
\end{figure}

There is a close relationship  between  1-planar graphs and right
angle crossing (RAC) graphs, where edges must be straight-line and
cross at a right angle \cite{del-dgrac-11}. 1-planar graphs and RAC
graphs have almost the same density, i.e., the maximal number of
edges for graphs of size $n$, namely $4n-8$ and $4n-10$. Eades and
Liotta \cite{el-racg1p-13} proved that every maximally dense RAC
graph is 1-planar. Conversely, every outer 1-planar graph has a RAC
drawing with the same embedding \cite{de-eogracd-12}. Hence, the RAC
graphs range between the outer 1-planar and the 1-planar graphs. In
fact,
outer 1-planar graphs are planar \cite{abbghnr-ro1pglt-13}.\\

Visibility representations have variously been generalized to two
dimensions
  with vertices  as non-overlapping paraxial rectangles
  and edges represented by  horizontal and  vertical
visibility. In the rectangle visibility approach
\cite{dh-rvrbg-94,hsv-orstt-95,hsv-rstg-99} horizontal and vertical
  edge-segments may cross and the resulting graphs have up
to $6n-20$ edges. Horizontal and vertical lines for edges were
allowed in   Biedl's flat visibility representation
\cite{b-sdogs-11}, however, the lines do not cross and the
horizontal lines are a shortcut for a local adjacency. Hence, the
concept is equivalent to weak visibility of planar graphs. The term
1- and 2-visibility was used by F\"{o}{\ss}meier et al.
\cite{fkk-2vdpg-96} for orthogonal drawings of planar graphs.

Dean et al. \cite{DEGLST-bkvg-07} introduced $k$-bar visibility,
where the vertices are represented as horizontal bars and bars are
allowed to see through at most  $k$ other bars. Thus $0$-bar
visibility is the common planar visibility, and  in $1$-visibility a
bar can be crossed by the visibility lines of many other bars. They
discussed the weak, $\epsilon$ and strong versions and showed that
$1$-bar visible graphs have at most $6n-20$ edges. In fact, the
formula for the density indicates that $k$-bar visible graphs are
related to $k$-quasi-planar graphs \cite{s-kqpg-12,FPS-nekqpg-13},
where no $k+2$ edges cross mutually. Recently,
 Sultana et al. \cite{srrt-b1vd1-13} showed that some special
 classes of graphs including the maximal outer 1-planar graphs
 are $1$-bar visible.

In this paper we generalize visibility representations such that
 they capture 1-planarity. The vertices are drawn as horizontal
 vertex-segments
 and an edge needs a vertical visibility and is represented by an edge-segment.
 Uncrossed segments
 are transparent and become impermeable if they are crossed by a segment of the other type.
  Hence, all crossings are right angle crossings (RAC) between a vertex
 and an edge, and each object is involved in at most one crossing.

 We show that every 1-planar graph has a 1-visibility
 representation in $\mathcal{O}(n^2)$ area, which can be computed in
 linear time from a given 1-planar embedding as a witness for
 1-planarity. This settles a conjecture of Sultana et al.
 \cite{srrt-b1vd1-13}.
 The algorithm uses the standard technique for visibility
 representations of planar graphs from \cite{dett-gdavg-99,rt-rplbopg-86, TT-vrpg-86}
 via the st-numbering
 of the graph and its dual, which operates on the planar skeleton without crossing edges.
  The given embedding is augmented  and transformed such
 that the 3-connected components have a
 normalized embedding \cite{abk-sld3c-13} and are separated by a
 copy of the edge between the separation pair. A local transformation
suffices to re-insert a pair  of crossing edges into the  face left
by their extraction. The
 3-connected components  are sandwiched between the horizontal
 vertex-segments of their separation pair, which comes directly from the st-numberings.

 1-visible graphs have the same maximal density as
 1-planar graphs with at most $4n-8$ edges for graphs of size $n$.
 This is readily seen, since a 1-visible graphs consists of a planar subgraph together
 with one crossing edge per vertex. Since the two
 outermost vertices are excluded the density reaches at most $4n-8$.
 So we provide a new and simple proof of the maximal density of
 1-planar graphs. The so-called  extended wheel graphs  $XQ_k$ \cite{bhw-1og-84} are examples
 of 1-planar graphs with maximal density.
The $XQ_8$ graph is shown in Fig. \ref{XQ8}, where the visibility
representation is obtained by our algorithm.
\begin{figure}
  \centering
  \subfigure {
      \includegraphics[scale=0.42]{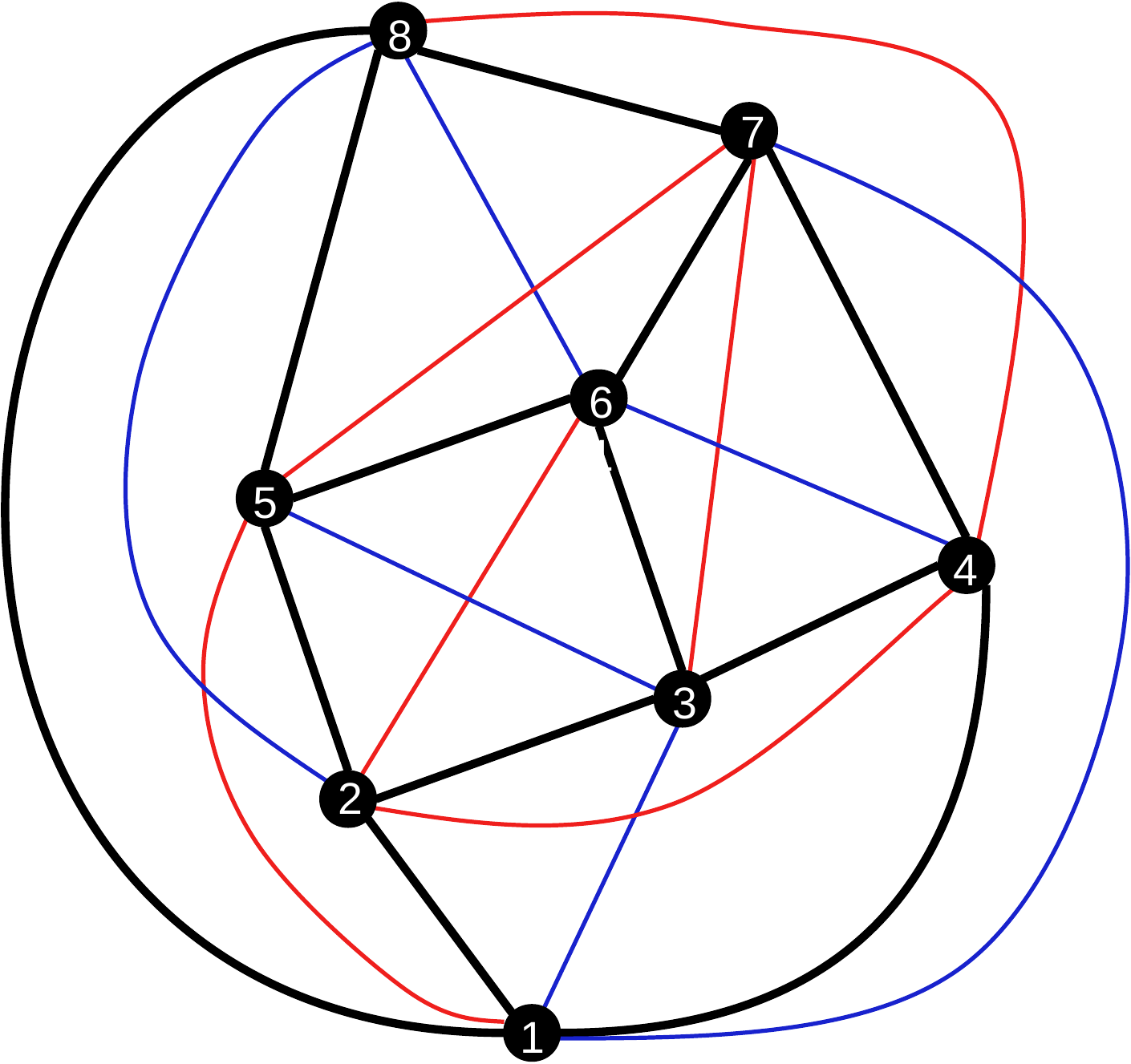}
      \label{fig:XQ8cross}
  }\quad\quad
  \subfigure {
    \includegraphics[scale=0.45]{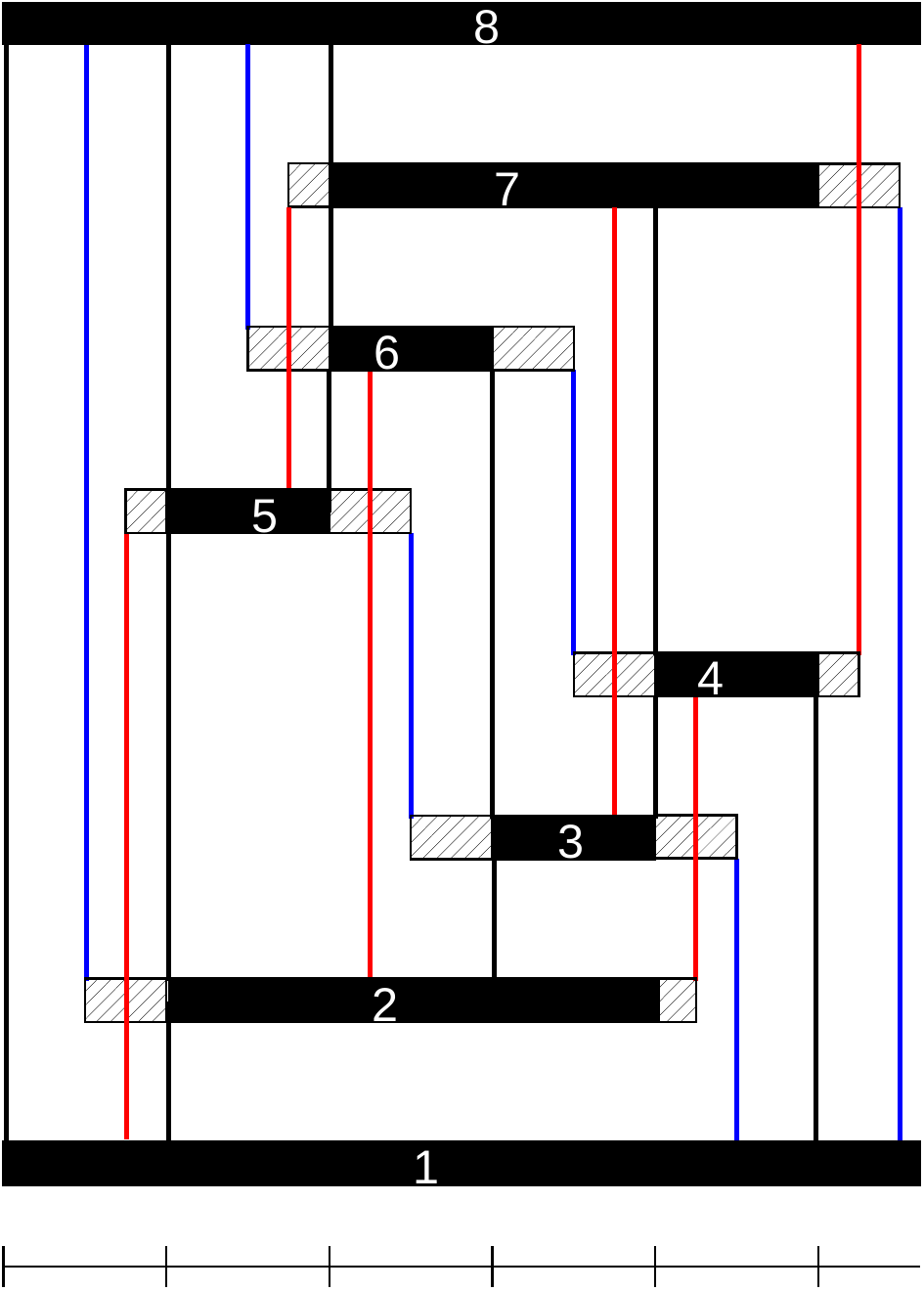}
    \label{fig:XQ8vis}
  }
  \caption{The extended wheel graph $XQ_8$ in common style and
  in visibility representation. The planar quadrilization is drawn with bold black lines,
  and each face has a red edge crossing a blue edge.
  \label{XQ8}}
\end{figure}
 However, there are 1-visibility graphs with $4n-8$ edges which are  not
 1-planar,
 including the complete graph on $7$ vertices without one edge, $K_7$-$e$,   which is not 1-planar
 \cite{bhw-1og-84,s-rm1pg-10}.
 Hence, the 1-visible graphs properly include the 1-planar graphs,
 even for maximally dense graphs.

\section{Preliminaries}

Consider simple undirected graphs $G = (V,E)$ with $n$ vertices and
$m$ edges. We suppose that the graphs are 2-connected, otherwise,
the components are treated separately, and are placed next to each
other as in \cite{TT-vrpg-86}. Note that  articulation points may
cause problems in visibility representations and make the difference
between weak and $\epsilon$-visibility. This difference vanishes if
the articulation points are in one face \cite{TT-vrpg-86}.
Articulation points do not matter in the  weak version of
visibility.

A \emph{drawing} of a graph is a mapping of $G$ into the plane such
that the vertices are mapped to distinct points and each edge is a
Jordan arc between its endpoints. A drawing is \textit{planar} if
the Jordan arcs of the edges do not cross and it is
\textit{1-planar} if each edge is crossed at most once. In 1-planar
drawings  crossings of edges with the same endpoint are excluded.

An embedding $\mathcal{E}(G)$ of a planar graph $G$ specifies
\emph{faces}. A face is a topologically connected region and is
given by a cyclic sequence of edges and vertices  that forms its
boundary. One of the faces is unbounded and is called the
\emph{outer face}.

Accordingly, a \emph{1-planar embedding} $\mathcal{E}(G)$ specifies
the faces in a 1-planar drawing of a graph $G$ including the outer
face. A 1-planar embedding is a witness for 1-planarity.  In
particular, it describes the pairs of crossing edges and the face
where the edges cross. Here a face is given by a cyclic list of
edges and half-edges and their vertices and crossing points. A
\emph{half-edge} is a segment of an edge from a vertex to a crossing
point. Each crossing point in a 1-planar embedding is incident to
four half-edges. If the crossing points are taken as new vertices
and the half-edges  as  edges, then we have the \emph{planarization}
of $\mathcal{E}(G)$, which is an embedded planar graph. This
structure is used by algorithms operating on  $\mathcal{E}(G)$,
where crossing points always remain as vertices of degree four and
may need a special treatment. Two 1-planar (planar) embeddings
$\mathcal{E}_1(G)$ and $\mathcal{E}_2(G)$ of a graph are
\emph{equivalent} if there is a homeomorphism $h$ on the plane with
$\mathcal{E}_2(G) = h(\mathcal{E}_1(G))$. Then one embedding can be
transformed into the other while preserving all faces including the
outer face. Such transformations are \emph{embedding preserving}.

A  \emph{visibility representation} of a planar graph displays the
vertices as horizontal bars, called vertex-segments, and two bars
must see each other along a vertical edge-segment if there is an
edge between the respective vertices. This is the \emph{weak}
version of visibility, where vertex-segments can see each other but
their vertices are not necessarily connected by an edge.

In this paper, we generalize visibility representations such that
they fit to 1-planar graphs. We use weak visibility, since we wish
to preserve the subgraph property: every subgraph of a 1-visible
graph is 1-visible. The vertex-segments include their extremes and
start and end at grid points. The $\epsilon$ and strong versions of
1-visibility do not seem useful, since many planar graphs cannot be
represented that way, such as circles of length at least four. The
if and only if condition between edges and 1-visibility enforces at
least one chord between non-adjacent vertices.

\begin{definition}
A 1-visibility representation of a graph $G=(V,E)$ displays each
vertex $v$ as a horizontal vertex-segment $\Gamma(v)$ and each edge
$e=(u, v)$ as a vertical edge-segment $\Gamma(e)$ from some point on
$\Gamma(u)$ to some point on $\Gamma(v)$.  The endpoints of all
segments are grid points. Vertex segments (edge segments) do not
overlap (in their interior).  Each vertex-segment is crossed by at
most one edge-segment and each edge-segment  may cross at most one
vertex-segment.
\end{definition}

Notice that 1-visibility drawings are straight-line drawings on
grids and there are right angle crossings between edges and
vertices. Hence, we have a new type of  RAC drawings
\cite{del-dgrac-11,el-racg1p-13}.

\section{Basic Properties}

A 1-planar embedding is \textit{planar maximal} if no further edge
can be added without inducing a crossing or multiple edges.  A
1-planar embedding can be augmented to a planar maximal embedding
via its planarization, where crossing points remain as vertices of
degree four. The augmentation can be computed in linear time from
the embedding. Note that the maximality depends on the embedding and
a different embedding of a graph may give rise to another maximal
planar augmentation, as the transformation of a B-configuration in
Fig. \ref{fig:B-configuration} into a X-configuration in Fig.
\ref{fig:X-configuration} illustrates. In X-configurations all four
vertices may have outer neighbors and there are at most three such
vertices in a B-configuration. Planar maximal embeddings have nice
properties.

\begin{lemma}
Let $\mathcal{E}(G)$ be a planar maximal 1-planar embedding.
\begin{enumerate}
  \item Every crossing induces a $K_4$ of the end vertices of the crossing edges.
  \item A face has  at most four  vertices, and there are
  such faces.
  \item Every (inner or outer) face is at most a $k$-gon with $k \leq 8$,
  where vertices and crossing points or alternatively half-edge
   are counted.
   \item A face has at most four crossing points.
\end{enumerate}
\end{lemma}

\begin{proof}
The first statement is due to the fact that  missing edges between
the end vertices can be routed near to the crossing edges. This has
been stated at several places, first of all in \cite{bhw-1og-84}.
For (2) the chords of a pentagon cannot be realized in a single
inner or outer face such that each chord has at most one crossing,
whereas a quadrangle can be realized as shown by  B-configuration
with an inner face with four vertices, see Fig.
\ref{fig:B-configuration}. Accordingly, there cannot be more than
$8$ half-edges in a face of a planar maximal 1-planar embedding, see
Fig. \ref{fig:face-8}, which also implies (4). \qed
\end{proof}

\begin{figure}
   \begin{center}
     \includegraphics[scale=0.4]{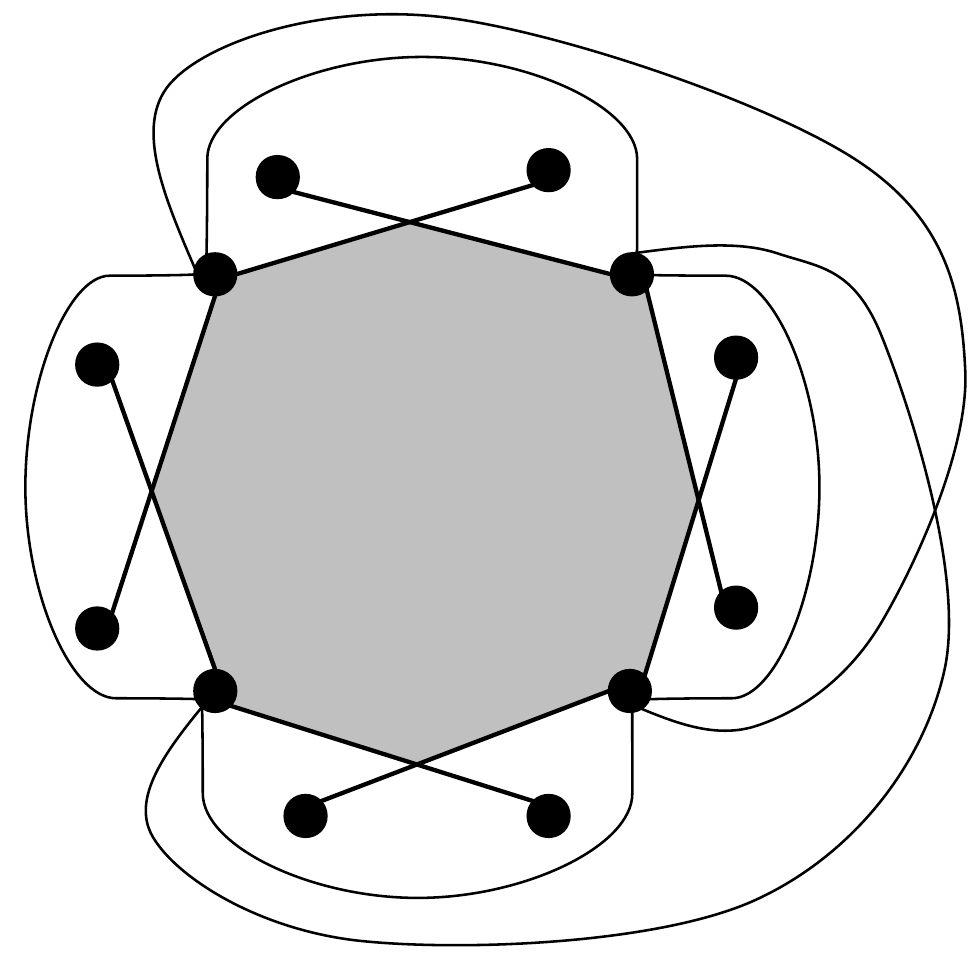}
     \caption{A face with 8 half-edges and four crossing points.
     \label{fig:face-8}}
   \end{center}
\end{figure}

Similar properties  were established
 for 1-planar embeddings without B- and W-configurations by Hong et al.
\cite{help-ft1pg-12} and by Alam et al. \cite{abk-sld3c-13} for
3-connected 1-planar graphs. Moreover, the faces can be simplified
if the embedding is changed. Then edges cross internally as in an
augmented X-configuration or externally as in a W-configurations,
and there
are at most two crossing points per face.\\

Eggleton \cite{e-rdg-86} raised the problem which  1-planar graphs
have drawings with straight-line edges. He solved this problem for
outerplanar graphs, where all crossing points are internal and
appear in X-configurations. Thomassen \cite{t-rdg-88}  characterized
the rectilinear 1-planar embeddings  by the exclusion of B- and
W-configurations, which are shown in Fig. \ref{fig:BWXconfig}.

\begin{definition}
Consider a 1-planar embedding.
 A \textit{B-configuration}
$B(u,v,p, x,y)$ consists of an edge $(u,v)$, called the base, and
two edges $(u, y)$ and $(v,x)$, which cross in some point $p$ such
that $x$ and $y$ lie in the interior of the triangle $(u, v, p)$.

A \textit{W-configuration} $W(u,v,p,p', x,y)$ consists of the base
edge $(u,v)$ and two pairs of edges $(u, y)$ and $(v, x)$ and $(u,
y')$ and $(v,x')$ which cross in points $p$ and $p'$ such that the
ends $x,y,x',y'$ lie in the interior of the quadrangle $u, p, v,
p'$.

A \textit{X-configuration} consists of a base edge and a pair of
crossing edges, where the crossing point lies in the interior of the
quadrangle of the endpoints of the edges.

A B- (W- and X-) configuration is \textit{augmented} if it contains
the probably missing  edges $(u,v), (u,x), (x,y), (y,v)$ and also
$(u, x'), (v,y'), (x', y')$ for augmented W-configurations.
\end{definition}

In the augmentations the edges cross in the inner face of a
X-configuration and in the outer face of a B-configuration. A
W-configuration comprises both.

 Note that the type of a configuration depends on the
embedding and the choice of the outer face or the routing of the
 base edge, which is drawn black and bold in Fig.
\ref{fig:BWXconfig}. A B-configuration becomes an X-configuration if
the inner and outer faces are exchanged, and vice-versa. In a
W-configuration the roles of the straight-line and curved crossing
edges swap by this exchange.
 This observation was used by Alam et al. \cite{abk-sld3c-13} in
their normal form theorem for embeddings of 3-connected 1-planar
graphs. Here, a given embedded 1-planar graph is first augmented by
planar edges to a planar maximal 1-planar graph and then the
embedding is transformed into normal form by local changes in the
cyclic order of the neighbors of some vertices. We recall this
result and stress the proof.

\begin{definition}
A planar maximal 1-planar embedding $\mathcal{E}(G)$ of a
3-connected 1-planar graph is in \textrm{normal form} if it has at
most one augmented W-configuration in the outer face, no augmented
B-configuration, and an augmented X-configuration does not contain a
vertex inside   the boundaries of the quadrangle of its endpoints.
\end{definition}

\begin{proposition} (Normal Form Theorem) \cite{abk-sld3c-13} \\
Let $G = (V,E)$ be a 3-connected 1-planar graph and $\mathcal{E}(G)$
a 1-planar embedding. There is a linear time algorithm to  transform
 $\mathcal{E}(G)$ into a planar maximal 1-planar embedding of a supergraph
 $H = (V, F)$ with $E \subseteq F$ such that
$\mathcal{E}(H)$ is in normal form.
\end{proposition}
\begin{proof}
For each pair of crossing edges $(a,c)$ and $(b,d)$ in
 $\mathcal{E}(G)$ add new edges $(a,b), (b,c), (c,d), (d,a)$
such that they are neighbors of the crossing edges at the end
vertices in  $\mathcal{E}(G)$. In other words, the edges are routed
along the crossing edges. If such an edge was already in $E$, then
remove it (the older) from $\mathcal{E}(G)$. This triangulates the
faces at crossing points. Thereafter, triangulate the planar faces.
These steps take  linear time, starting from $\mathcal{E}(G)$. The
result is the embedded supergraph $H$. Now all B-, W-, and
X-configurations are augmented. Each augmented B-configuration which
is not a W-configuration is transformed into a X-configuration by
the re-routing of the base. Two B-configurations on opposite sides
of the base and connected by an edge crossing the base are merged to
a W-configuration. There is no vertex inside the boundaries of the
end vertices of an X-configuration. $\mathcal{E}(H)$ cannot contain
two augmented W-configurations or a W-configuration in its interior,
since the base of a W-configuration is a separation pair, which is
excluded by 3-connectivity. Hence, $\mathcal{E}(H)$ is planar
maximal by the augmentation and triangulation.
 \qed
\end{proof}

The normal form theorem holds for every 3-connected component of a
1-planar graph $G$. Suppose that  $G$ is 2-connected with an
embedding $\mathcal{E}(G)$ with planar maximal 3-connected
components in normal form. For every separation pair $\{ u,v \}$
there is a sequence of 3-connected 1-planar graphs $C_0,\ldots,
C_{k-1}$ in clockwise order at $u$, and each pair of adjacent
components $C_i$ and $C_{i+1}$ with $0 \leq i \leq k-1$ is separated
by a pair of crossing edges from a W- or an X-configuration or both.
Otherwise, such components merge to a single planar maximal
3-connected component. If the separation pair is in the outer face
and there is no a pair of crossing edges from a W-configuration in
the outer face, then the outermost copy $e_k$ can be saved.

To separate the components at a separation pair $\{u.v\}$ even
further we allow multiple edges and introduce copies $e_i$ for $i=1,
\ldots, k$ of the edge $e_0 = (u,v)$  as \emph{separation edges}.
The $i$-th separation edge $e_i$ is routed next to a pair of
crossing edges which separates $C_{i-1}$ from $C_i$.If present the
outermost separation edge $e_k$ encloses all components and the
multi-edges $e_0, e_k$ form the outer face. This situation also
holds relative to a separation pair.

For a counting argument each separation edge can be taken for an
edge between the components it separates or from a crossing point to
a  vertex or another crossing point in the planarization.

All steps for the augmentation with 3-connected components in normal
form and separation edges take linear time on $\mathcal{E}(G)$. Thus
we can state.

\begin{lemma}
Every 1-planar embedding $\mathcal{E}(G)$ can be transformed in
linear time into a planar maximal 1-planar embedding
$\mathcal{E}(G')$ of a supergraph $G'$. $G'$ is a graph with
multi-edges, each 3-connected component of $\mathcal{E}(G')$ is in
normal form and there is a separation edge between  adjacent
$3$-connected components at a separation pair $\{ u,v \}$.
\end{lemma}

This situation is depicted in Fig. \ref{Fig:separation}, where the
copies are the edge $e_0 = \{u,v\}$ is drawn dotted and blue.

\begin{figure}
   \begin{center}
     \includegraphics[scale=0.55]{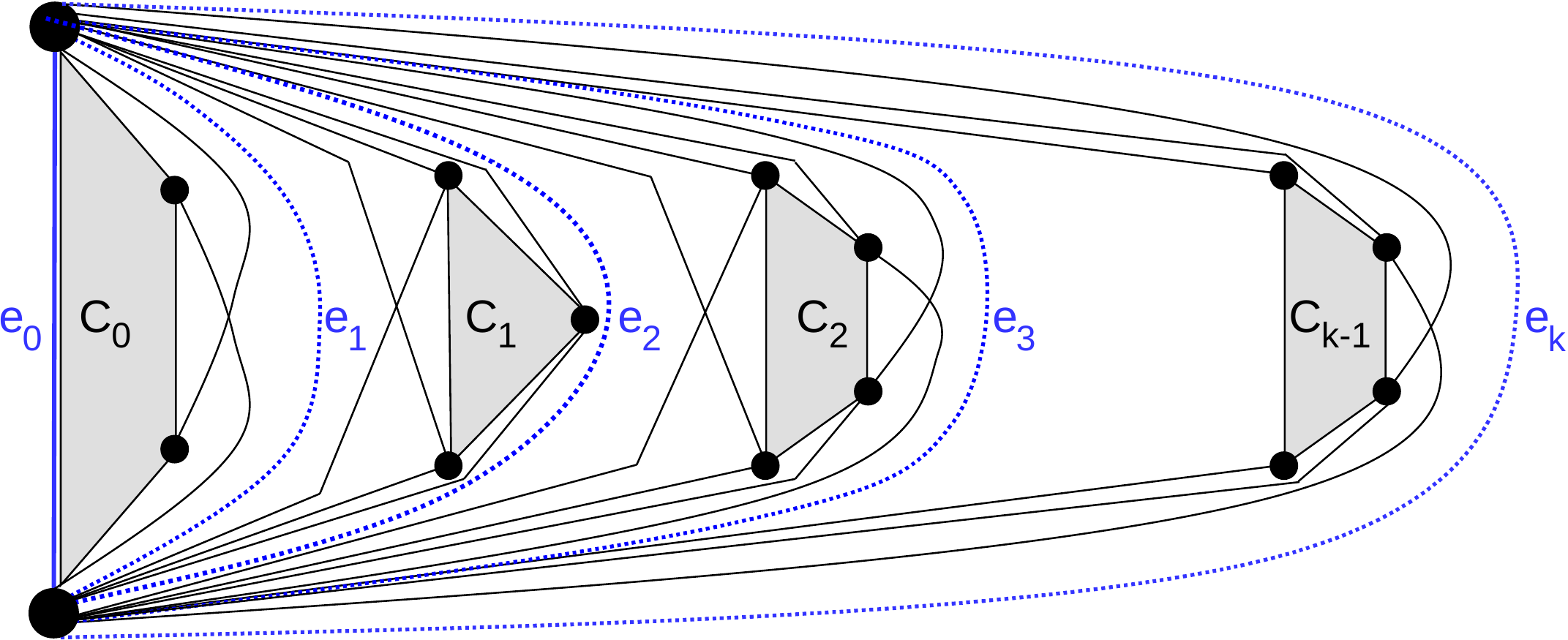}
     \caption{A sequence of planar maximal 1-planar graphs at a separation pair $\{u,v\}$}
     and separating edges.
     \label{Fig:separation}
   \end{center}
\end{figure}

Hong et al. \cite{help-ft1pg-12} showed that a 1-planar embedding
can be transformed into a straight-line 1-planar drawing, which
preserves the embedding, provided there are no B- and
W-configurations. Their algorithm is quite complex and uses the
SPQR-tree data structure for the decomposition of the graph into its
3-connected components and the convex drawing algorithm for planar
graphs from \cite{cyn-lacdpg-84}, which needs a high resolution for
its numerical computations. There is no stated bound on the area,
but it is likely to be exponential. However, each augmented B- and
W-configuration induces one edge with a bend if the embedding is
preserved. Hence, a straight-line drawing of a 1-planar graph may
have a linear number of edges with a bend. The sparse maximal
1-planar graphs  from Fig. 3 in \cite{begghr-odm1p-13} may serve as
an example.

Here we capture all 1-planar graphs  and provide a 1-visibility
representation with straight vertical lines for all  edges.

\section{Visibility Representation}

In this section we show that every 1-planar graph has a 1-visibility
drawing. The result is obtained by the 1-VISIBILITY algorithm, whose
input is an embedding $\mathcal{E}(G)$ as a witness for 1-planarity.
After a planar maximal augmentation it considers each 3-connected
component $C$, transformes $C$ into normal form, and separates
3-connected components at a separation pair by separation edges.
Then the graph and in particular each 3-connected component is
planarized by the extraction of the pairs of crossing edges. The
normal form and the separation edges guarantee that each face has at
most one pair of crossing edges. The so obtained planar graph is
drawn by the common planar visibility algorithm. Thereafter,
CROSSING-INSERTION  reinserts each pair of crossing edges in the
face from which is was extracted. Finally, the edge-segments of
added edges are hidden.

Consider a  planar visibility algorithm  from
\cite{dett-gdavg-99,rt-rplbopg-86, TT-vrpg-86}.  It  takes an
embedded planar graph and two vertices $s,t$ in the outer faces and
directs the edges according to an $st$-numbering from $s$ to $t$.
Thereafter each vertex $v$ except $s,t$ has a neighbor with a
smaller and a larger st-number than itself and two sub-sequences of
incoming and outgoing edges. In other words, each vertex and $G$ are
bi-modal \cite{rt-rplbopg-86}. Route the edge $(s,t)$ to the left of
the drawing of $G$. Then consider the directed dual $G^*$, where
$s^*$ is the face to the right of the $(s,t)$ edge (or the left half
of the outer face) and $t^*$ is (the right half of) the outer face,
and direct its edges according to the $s^*t^*$-numbering of $G^*$.
Recall that $G$ was extended by separation edges between
$3$-connected components, which has an impact on $G^*$.

Define the  \emph{distance} $\delta(v)$ of a vertex $v$ by its
st-number as in \cite{rt-rplbopg-86, TT-vrpg-86} or for a more
compact drawing \cite{dett-gdavg-99} by the length of a longest path
from $s$ and accordingly define the dual distance $\delta^*(f)$ of a
face $f$ in $G^*$.  Then $\delta(s)=0$, $\delta(t)=h-1$,
$\delta^*(s^*)=0$ and $\delta^*(t^*) = w-1$ for some $h \leq n$ and
$w \leq 2n-5$ and the visibility representation is of size $w \times
h$. The insertion of separation edges does not affect the upper
bound of $2n-5$, since for each separation edge $e_i$ there is at
least one edge missing from $C_i$ to the next component $C_{i+1}$ in
cyclic order. For the compacted version one must take care that the
distance is different for vertices $b$ and $d$ of a quadrangle $f =
(a,b,c,d)$, whose bottom and top are $a$ and $c$ if there is an
augmented X-configuration. The  requirements are met by the
st-number and can otherwise be achieved by a local lifting as in
\cite{b-sdogs-11}. Moreover, if $\{u,v\}$ is a separation pair with
a sequence of $3$-connected components $C_0,\ldots, C_{k-1}$ in
clockwise order at $u$ and separation edges $e_0, \ldots, e_k$ and
the st-number of $u$ is smaller than the st-number of $v$, then the
st-numbering implies that $\delta(u) < \delta(w) < \delta(v)$ for
every vertex $w$ from any component $C_i$ and $\delta^*(e_{i-1}) <
\delta^*(f) < \delta^*(e_i)$ if $f$ is an inner face of $C_{i-1}$
and $\delta^*(e_i)$ is the dual distance of the face immediately to
the left of $e_i$.

For each edge $e=(u, v)$ let $left(e)$ ($right(e)$) be the dual
distance $\delta^*(f)$ of the face $f$ of $G$ to the left (right) of
$v$ and let $left(v)$ ($right(v)$) be the least (largest) dual
distance
of a face incident with $v$.\\


\begin{algorithm}
  \caption{PLANAR-VISIBILITY}\label{alg:planar-visibility}

  \KwIn{A 2-connected planar graph $G$ (with multi-edges) with a planar embedding $\mathcal{E}(G)$.}

  \KwOut{A visibility representation $\mathcal{VR}(G)$.}

  Construct an st-numbering of  $G$ with $(s,t)$ on the left.\;
  Compute the dual graph  $G^*$.\;
  Compute the distance   $\delta(v)$ for all vertices $v$ of $G$  and the dual
  distance $\delta^*(f)$ for all faces $f$.\;
  \ForEach{vertex~$v$ of $G$}{%
    draw the vertex-segment $\beta(v)$ at the $x$-coordinate
    $\delta(v)$ from $\delta^*(left(v))$ to $\delta^*(right(v))-1$
    if $v \neq s,t$, and from $0$ to $\delta^*(t)$ for $s,t$.
  }
  \ForEach{edge $e=(u,v)$ of $G$}{%
    draw a vertical edge-segment between $(\delta^*(left(e),
    \delta(u))$ and  $(\delta^*(left(e), \delta(v))$
  }
\end{algorithm}

 The correctness of  PLANAR-VISIBILITY and the linear
running time was proved in
\cite{dett-gdavg-99,rt-rplbopg-86, TT-vrpg-86}.\\

We use PLANAR-VISIBILITY to draw 3-connected components $C_i$ of
1-planar graphs, whose pairs of crossing edges $(a,c)$ and $(b,d)$
are first extracted  and are then  reinserted in the face they left
behind. The normal form embedding and the added separation edge
$e_i$ to the right of $C_i$ guarantee that each pair of crossing
edges has its own face $f$, which is a quadrangle. $f$ comes from an
augmented X-configuration if it is an inner face or is the relative
outer face of a W-configuration and is immediately to the left of
$e_i$, where   $e_i$ is a separation edge.

For a face $f=(a,b,c,d)$ let
 $a$ be the lowest vertex in the visibility drawing of PLANAR-VISIBILITY,
 i.e., the y-coordinate $\delta(a)$ is minimal. We call
$f$ a \emph{left-wing} (\emph{right-wing}) if $\delta(a) < \delta(b)
< \delta(c) < \delta(d)$ and $b,c$  are to the left (right) of  $f$,
and a \emph{diamond} if $\delta(a) < \delta(b), \delta(d) <
\delta(c)$. $f$ is a left-wing if $f$ is the outer face or if
$(a,d)$ is a separation edge.

There are always two options, which of the two middle vertices of a
quadrangle  $f$ is crossed by an edge. A maximal bipartite matching
determines one vertex per face and guarantees that each vertex is
crossed at most once.

The crossing insertions are illustrated in Fig.
\ref{fig:crossinginsertion}.


\begin{algorithm}
  \caption{CROSSING-INSERTION}\label{alg:crossing-insertion}
  \KwIn{%
    A visibility representation of a face $f$ with the vertices
    $(a,b,c,d)$, where $a$ has the lowest $y$-coordinate $y(a)$, and a
    pair of edges $(a,c)$ and $(b,d)$ crossing in $f$, such that the
    vertex-segment of $b$ is crossed by the edge-segment of $(a,c)$.
    (The case where the other inner vertex is crossed is similar).
  }
  \KwOut{A 1-visibility representation of $f$ with $(a,c)$ crossing $b$.}

  \Switch{type of face~$f$}{%
      \Case{  $f$ is a left-wing}{%
        extend $\beta(b)$ by 0.5
        and  $\beta(c)$ by 0.25 units to the right and
        draw $(a,c)$ at the $x$-coordinate $\delta^*(f)-0.75$ and
        $(b,d)$ at $\delta^*(f)-0.5$
        }\Case{$f$ is a right-wing}{%
        extend $\beta(b)$ by 0.5 and  $\beta(c)$ by 0.25 units
        to the left   and draw $(a,c)$ at the $x$-coordinate
        $\delta^*(f)-0.25$ and
        $(b,d)$ at $\delta^*(f)-0.5$
        }\Case{$f$ is a diamond}{%
        extend $\beta(b)$ by $0.5$ units to the right
        and $\beta(d)$ by $0.5$ units to the left, draw $\beta(b,d)$
        at $\delta^*(f)-0.5$ and draw $\beta(a,c)$ at $\delta^*(f)-0.75$
        if $b$ is crossed and at $\delta^*(f)-0.25$ if  $d$ is crossed.
      }
  }
\end{algorithm}

\begin{figure}
  \centering
  \subfigure[left-wing]{
      \includegraphics[scale=0.5]{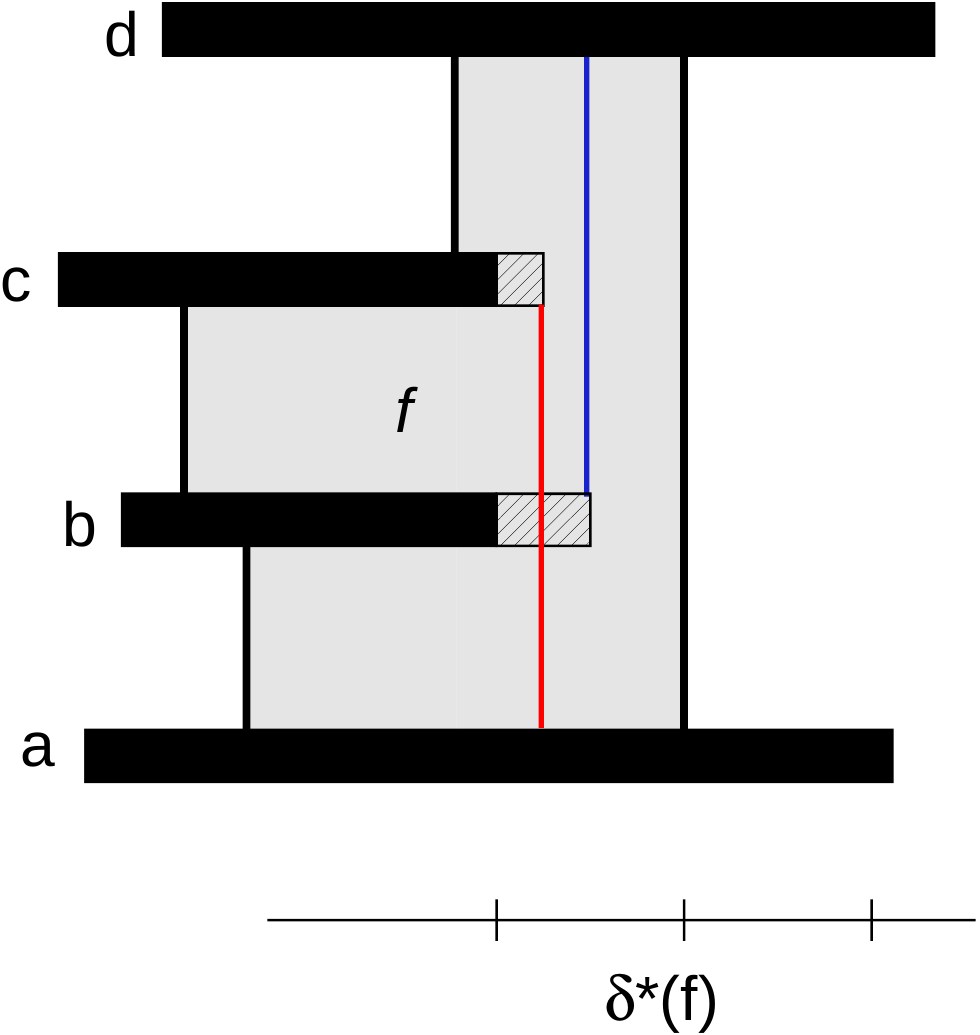}
      \label{fig:leftwing}
  }\quad\quad
  \subfigure[diamond]{
    \includegraphics[scale=0.5]{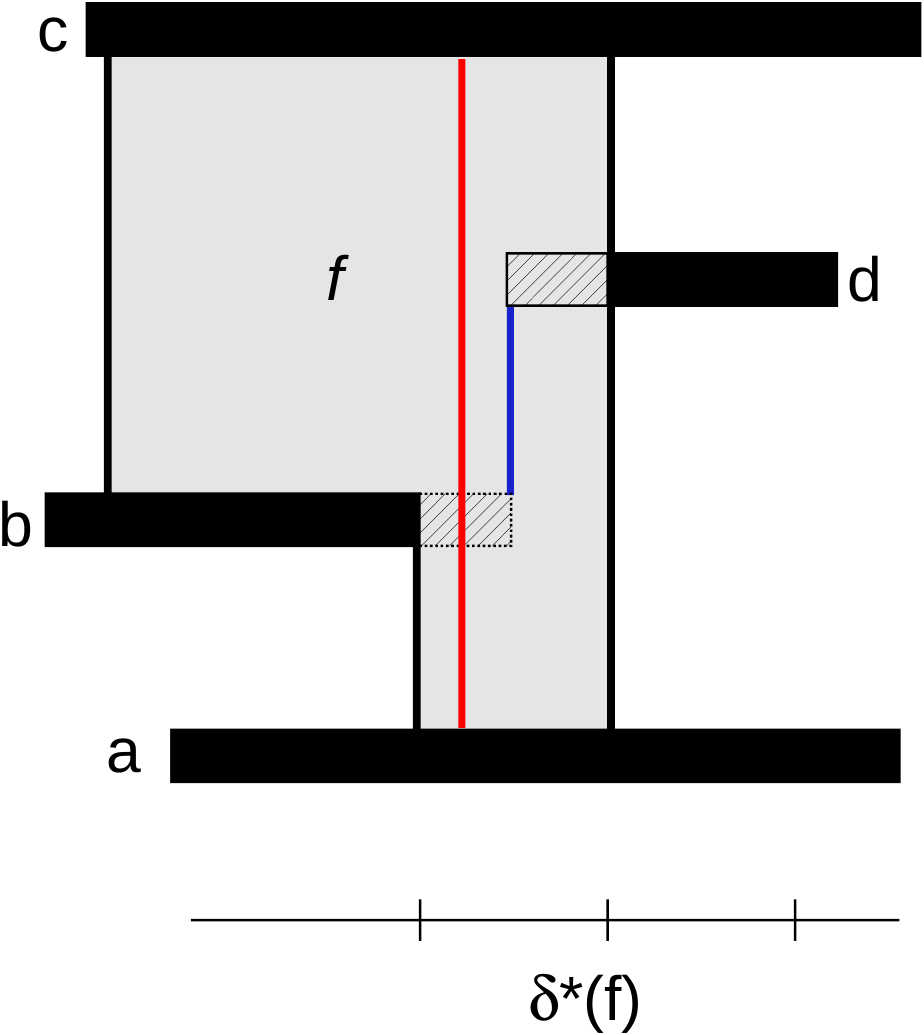}
    \label{fig:outerface}
  }
  \caption{Reinserting a pair of crossing edges $(a,c)$ and $(b,d)$ in face $f$ by an extension
  of the inner vertex-segments, where $(a,c)$ crosses  $b$.
  \label{fig:crossinginsertion}}
\end{figure}

\begin{lemma} \label{proofCrossingInsertion}
If a face $f = (a,b,c,d)$ is drawn by PLANAR-VISIBILITY, then
CROSSING-INSERTION adds the pair of crossing edges $(a,c)$ and
$(b,d)$ inside $f$ with exactly one vertex-edge crossing.
\end{lemma}

\begin{proof}
If $f$ is  a left-wing, then the vertex-segments of $b$ and $d$ end
at $\delta^*(f)-1$ and the edge-segments are at or to the left of
$\delta^*(f)-2$. The edge-segment of $(a,d)$ is at or to the right
of $\delta^*(f)$. Hence, the extension of $\beta(b)$ and $\beta(c)$
does not intersect the edge-segment of $(a,d)$. The edges $(a,c)$
and $(b,d)$ are routed inside $f$ and induce a crossing of $b$ and
$(a,c)$. The case where $f$ is a right-wing is symmetric. Then the
edge-segment of $(a,d)$ is at or to the left of $\delta^*(f)-1$, and
the edge-segments $(a,b), (b,c), (c,d)$ are right aligned at
$\delta^*(f)$. The vertex-segments of $b$ and $c$ begin at
$\delta^*(f)$. If $f$ is a diamond with $b$ on the left and $d$ on
the right, then $\beta(b)$ ends at $\delta^*(f)-1$ and $\beta(d)$
begins at $\delta^*(f)$, and the $y$-coordinates of $b$ and $d$ are
different, since the distance $\delta$ guarantees this property.
Again there is
 a single
vertex-edge crossing in $f$. The vertex-segments of the extreme
vertices cover the range from $\delta^*(f)-1$  to $\delta^*(f)$, and
generally go far beyond. \qed
\end{proof}

Finally, consider a separation pair $\{u, v \}$ and its 3-connected
$C_0, \ldots, C_{k-1}$, which are separated by separation edges
$e_1, \ldots, e_k$ as copies of $e_0 = (u,v)$. Associate $e_i$ with
$C_i$ as its base. Then the 3-connected components are sandwiched
between the vertex-segments of $u$ and $v$ and two adjacent
components $C_{i-1}$ and $C_i$ are clearly separated by $e_i$ in a
left-to-right order, which is due to the $st$- and
$s^*t^*$-numberings.


\begin{algorithm}

  \caption{1-VISIBILITY}\label{alg:1-visibility}

  \KwIn{An embedded 2-connected 1-planar graph $\mathcal{E}(G)$.}
  \KwOut{A 1-visibility representation $\mathcal{VR}(G)$ on a grid of quadratic size.}

  Augment $\mathcal{E}(G)$ to a planar maximal 1-planar embedding $\mathcal{E}(G')$,
  e.g., via a maximal planar augmentation of its planarization which keeps the crossing
  points at degree four.\;
  Decompose $G'$ into its 3-connected components.\;

  \ForEach{separating pair $\{u,v\}$}  {%
  compute the sequence of the 3-connected components
  $C_i$ for $i= 1, \ldots, k$ and add a copy $e_i$ of $(u,v)$
  as a separation edge to the right of $C_{i-1}$.
  }

  If the embedded graph has a crossing in the outer face, then
  add a copy of the base edge as a separation edge to cover the crossing from the outer face.
  Let $G''$ be the intermediate graph. \;
  Transform the embedding of each 3-connected component of $G''$ into normal
  form.\;
  Planarize $\mathcal{E}(G'')$ to $\sigma (\mathcal{E}(G''))$
     by the extraction of  all pairs of crossing  edges.\;

     Construct a visibility representation of $\sigma (\mathcal{E}(G''))$
     by PLANAR-VISIBILITY.\;

     (Separately for each 3-connected component) Compute the set of crossed vertex-segments by a
     maximum bipartite matching on the set of faces $F$ including a pair of crossing
     edges and the set   $I$ of inner vertices of the faces of $F$.\;

     Reinsert the crossing edges by CROSSING-INSERTION.\;

    Scale all x-coordinates by the factor $4$. \;

    Ignore or hide the edges from the augmentations to $G'$ and $G''$.

\end{algorithm}

We can now establish our main result.

\begin{theorem}

There is a linear time algorithm to construct a 1-visibility
representation of an embedded 1-planar graph on a grid of size at
most $(8n-20) \times (n-1)$.
\end{theorem}

\begin{proof}
First consider the case where the graph $G$ is 3-connected. Its
embedding is transformed into normal form with all crossings as
augmented X-configurations with the exception of at most one
crossing in the outer face. Now each crossing of a pair of edges has
its own face, where a crossing in the outer face is assigned to the
face to the left of the inserted separation edge, and each such face
is a quadrangle. This property also holds for 2-connected graphs by
the separation edges between 3-connected components. Hence, the
planar graph  after the extraction of all pairs of crossing edges
can be drawn by PLANAR-VISIBILITY, and the extracted edges can be
reinserted by CROSSING-INSERTION. This induces the crossing of a
single vertex-edge pair for each pair of crossing edges in $f$, as
shown in Lemma \ref{proofCrossingInsertion}, such that each edge is
crossed at most once.

Multiple vertex crossings are excluded by a maximum matching between
the set of faces $F$ with a crossing and the set of
 inner vertices $I$ associated with the faces of $F$.
By the st-numbering each vertex  $v$ is an inner vertex  of at most
two faces, one to the left and one to the right. $v$ can be the top
or bottom vertex of  other faces. Hence, $v$ is assigned to at most
two faces of $F$, and each $f \in F$ has two inner vertices, as can
be seen from the left-wing, right-wing or diamond shape. The maximum
bipartite matching problem over $F$ and $I$ has a solution by Hall's
marriage theorem \cite{h-matching-35}, since for every subset $F'
\subseteq F$ the number of inner vertices $|I'|$ of the faces from
$F'$ is greater or equal to $|F'|$.

In this particular case, a maximum matching can be computed in
linear time by first matching all inner vertices of degree one, and
then matching the remaining faces using at most one alternation.
Since the remaining faces and inner vertices all have degree two,
the bipartite graph decomposes into disjoint alternating cycles.

PLANAR-VISIBILITY computes grid points for the segments and uses an
area of at most $(2n-5) \times (n-1)$ including the separation
edges. The number of faces of the augmented graph $G''$ is bounded
from above by $2n-4$, since for each separation edge there is a
missing edge between the adjacent 3-connected components.
CROSSING-INSERTION does not increase the area, but needs a scaling
of the $x$-coordinates by four, which results in an area of at most
$(8n-20) \times (n-1) $.

All steps take linear time. Steps 1-4, 7, 11 and 12 are done on
planar graphs. The linear running time of step 6 is from
\cite{abk-sld3c-13} and of step 8 from
\cite{dett-gdavg-99,rt-rplbopg-86, TT-vrpg-86}. Step 10 takes
$\mathcal{O}(1)$ time per crossing, and there are at most $n-2$
crossings, and step 5 is a single action. Finally, step 9 is shown
above. \qed

\end{proof}

\begin{corollary} \label{subclass}
Every 1-planar graph is a  1-visibility graph.
\end{corollary}

\section{Density}

It is easily seen that 1-visibility graphs of size $n$ have at most
$4n-8$ edges, since there are at most $3n-6$ planar edges and at
most $n-2$ edges which cross a vertex. This is exactly the upper
bound of the density of 1-planar graphs.

\begin{corollary}
A 1-visibility graph of size $n$ has at most $4n-8$ edges.
\end{corollary}

From Corollary \ref{subclass} we obtain a new and simple proof for
the maximal density of 1-planar graphs, which was proved before in
\cite{bhw-bs-83,pt-gdfce-97,fm-s1pg-07}.

\begin{corollary}
A 1-planar graph  of size $n$ has at most $4n-8$ edges.
\end{corollary}

Surprisingly, there are 1-visible graphs which are not 1-planar,
even if they have the maximum of $4n-8$ edges.

\begin{theorem}
For every $n \geq 7$ there are graphs with $4n-8$ edges which are
1-visible and not 1-planar.
\end{theorem}

\begin{proof}
There are no 1-planar graphs with $n=7$ (or $n=9$) vertices and
$4n-8$ edges \cite{bhw-1og-84,s-rm1pg-10}, however, the complete
graph on $7$ vertices without one edge $K_7$-$e$ is 1-visible, as
shown in Fig. \ref{fig:K7-e}.

For $n \geq 8$ construct the graph  $G_n$ from $K_7$-$e$ and add
$n-7$ vertices and connect each such $v_i$ with vertex $3$ on the
left and with vertex $1$ on the right side and with $v_{i-1}$ and
$v_{i-2}$ on top, where the edge $(v_i, v_{i-2})$ crosses $v_{i-1}$,
as illustrated in Fig. \ref{fig:K7-e}.

Since the 1-planar graphs have the subgraph property  $G_n$ is not
1-planar.
 \qed
\end{proof}

\begin{figure}
   \begin{center}
  \includegraphics[scale=0.45]{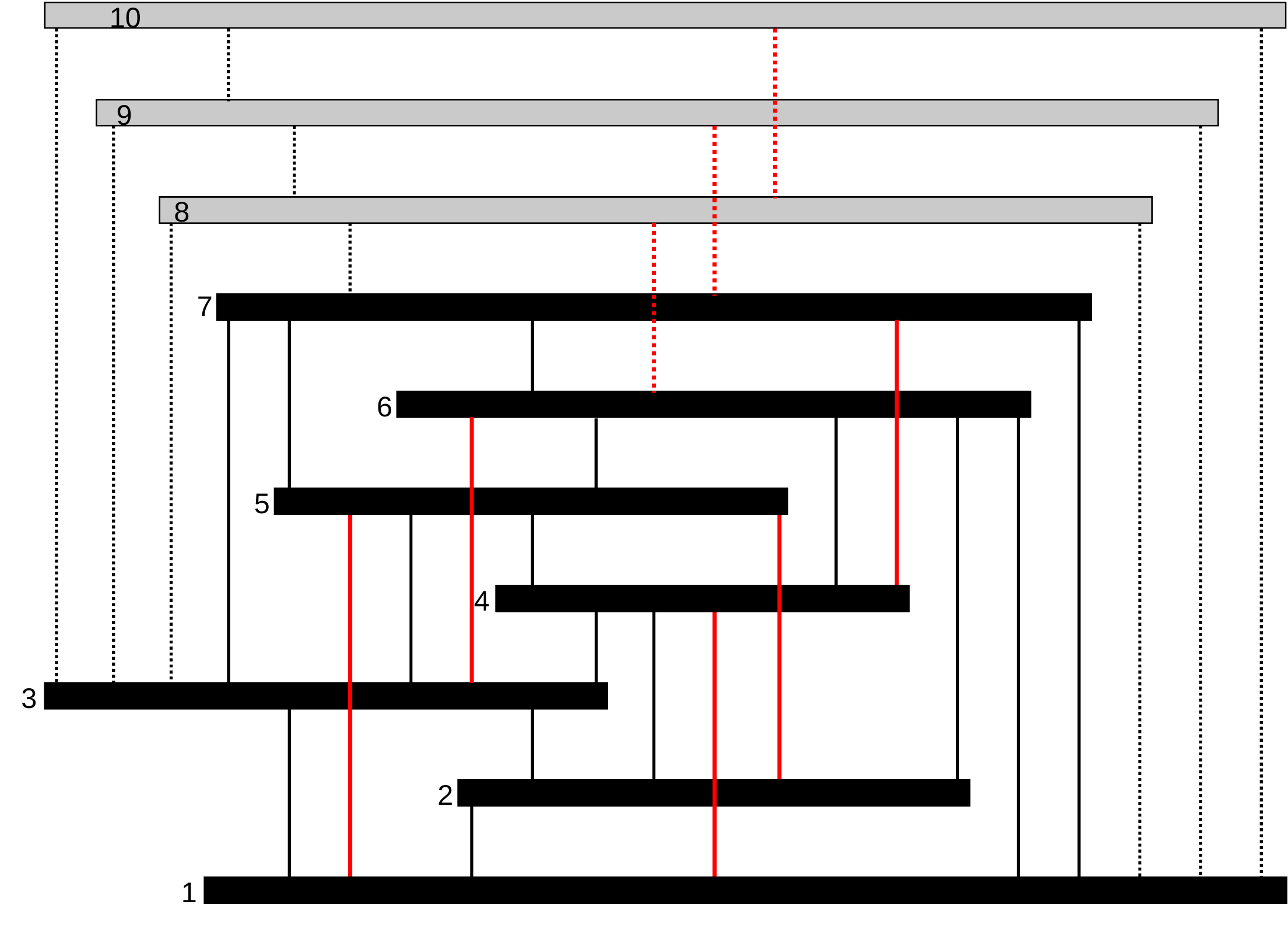}
     \caption{The $K_7$-$e$ graph with the vertices $\{1,\ldots, 7 \}$ is 1-visible and not 1-planar.
     The edge $(2,7)$ is missing. The graph can be expanded by new vertices $8,\ldots$ which add four edges
     to lower vertices.}
     \label{fig:K7-e}
   \end{center}
\end{figure}

1-planar (1-visible) graphs with   $4n-8$ edges  are  called
\emph{optimal} \cite{bhw-1og-84,s-rm1pg-10}. Note that there are
optimal 1-planar graphs only for $n=8$ and $n \geq 10$
\cite{bhw-1og-84,s-rm1pg-10}, whereas there are optimal 1-visible
graphs for every $n \geq 7$. More 1-visible and not 1-planar graphs
can be constructed using the schema of Fig. \ref{fig:wall}, where
the outer frame represents a subgraph with a unique 1-planar
embedding as in \cite{km-mo1ih-13} and the edge $(a,c)$ crosses
vertex $b$ and would cross at least two edges in every 1-planar
drawing.

\begin{figure}
   \begin{center}
   \includegraphics[scale=0.45]{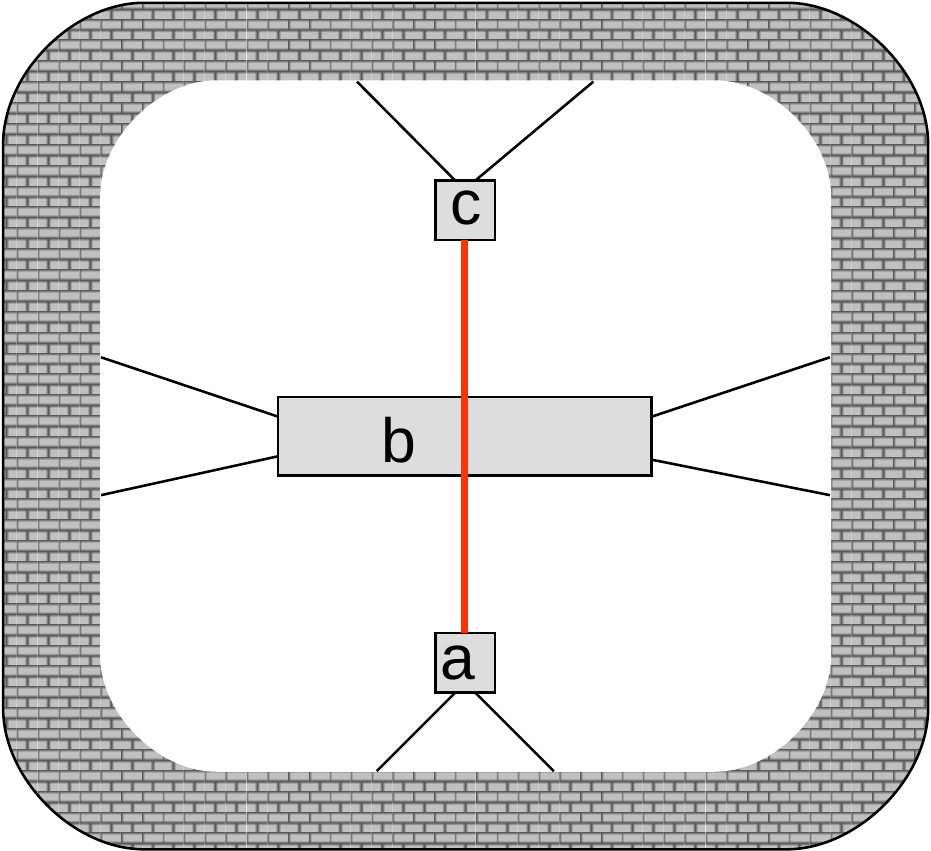}
     \caption{Schema  for non 1-planar 1-visibility graphs}
     \label{fig:wall}
   \end{center}
\end{figure}

\section{Conclusion and Perspectives}

We introduced 1-visibility drawings as a novel crossing style. If
this is restricted to a single crossing per object, then edge-vertex
crossings properly extend the edge-edge crossings.

The new crossing style raises several questions.
\begin{enumerate}
  \item Is the recognition problem for 1-visible graphs \NP \, complete?
  Recall that recognizing 1-planar graphs is \NP \, complete
  \cite{km-mo1ih-13}, even under restrictions
  \cite{abgr-o1pgr-12,km-mo1ih-13}. How hard is it to test
  whether a 1-visible graph is 1-planar?
  \item When does a graph have a unique 1-visibility representation?
  We're looking for a parallel to Whitney's theorem on unique
  embeddings of 3-connected planar graphs on the sphere.
  \item We studied the weak version for 1-visibility. What are the
  restrictions imposed by the $\epsilon$ and strong versions of
  1-visibility with an if and only if relation between
  1-visibility and edges?
  \item Study a generalization of 1-visibility that is based on Biedl's flat visibility
  representation and permits a single  crossing between
  horizontal bars for vertices or edges and vertical lines for
  edges? How close is this approach to rectangle visibility
  graphs \cite{dh-rvrbg-94,hsv-orstt-95,hsv-rstg-99}?
  \item Consider $k$-visibility for $k \geq 1$.
\end{enumerate}

\section{Acknowledgement}
I wish to thank K. Hanauer for providing a 1-visibility drawing of
the  $K_7$-$e$ graph.


\end{document}